\documentclass[superscriptaddress,11pt,nofootinbib,standalone]{revtex4-1}

\usepackage{bbm,bm,mathtools,slashed}
\usepackage{graphicx,color}
\usepackage{hyperref}
\usepackage[T1]{fontenc}
\usepackage{enumerate}
\usepackage[all]{xy}
\usepackage{amsmath,amssymb,amsthm}
\usepackage{subcaption}
\usepackage{ascmac}
\usepackage{braket}
\usepackage{extarrows}
\usepackage{cases}
\usepackage{tikz}
\usetikzlibrary{intersections, calc}

\newcommand{\Disp}[1]{{\displaystyle{#1}}}

\newcommand{\ot}{\otimes}
\newcommand{\bigot}{\bigotimes}

\newcommand{\bpmat}{\begin{pmatrix}}
\newcommand{\epmat}{\end{pmatrix}}

\newcommand{\init}{\mathrm{init}}
\newcommand{\ter}{\mathrm{ter}}

\newcommand{\rank}{\mathrm{rank}\,}
\newcommand{\cprod}{\,\Box\,}

\newcommand{\setmid}{\mathrel{}\middle|\mathrel{}}

\newtheorem{definition}{Definition.}
\newtheorem{theorem}{Theorem.}

\begin{document}

\title{
Equivalence of lattice operators and graph matrices
}

\author{Jun Yumoto}
\email{d8521007(at)s.akita-u.ac.jp}
\address{Department of Mathematical Science, Akita University, Akita 010-8502, Japan}

\author{Tatsuhiro Misumi}
\email{misumi(at)phys.kindai.ac.jp}
\address{Department of Physics, Kindai University, Osaka 577-8502, Japan}
\address{Research and Education Center for Natural Sciences, Keio University, Kanagawa 223-8521, Japan}

\begin{abstract}
We explore the relationship between lattice field theory and graph theory, placing special emphasis on the interplay between Dirac and scalar lattice operators and matrices within the realm of spectral graph theory. Beyond delving into fundamental concepts of spectral graph theory, such as adjacency and Laplacian matrices, we introduce a novel matrix named as "anti-symmetrized adjacency matrix", specifically tailored for cycle digraphs ($T^1$ lattice) and simple directed paths ($B^1$ lattice). 
The nontrivial relation between graph theory matrices and lattice operators shows that the graph Laplacian matrix mirrors the lattice scalar operator and the Wilson term in lattice fermions, while the anti-symmetrized adjacency matrix, along with its extensions to higher dimensions, are equivalent to naive lattice Dirac operators.
Building upon these connections, we provide rigorous proofs for two key assertions:
(i) The count of zero-modes in a free lattice scalar operator coincides with the zeroth Betti number of the underlying graph (lattice).
(ii) The maximum count of Dirac zero-modes in a free lattice fermion operator is equivalent to the cumulative sum of all Betti numbers when the $D$-dimensional graph results from a cartesian product of cycle digraphs ($T^1$ lattice) and simple directed paths ($B^1$ lattice).
\end{abstract}

\maketitle

\newpage

\tableofcontents

\newpage


\section{Introduction}
\label{sec:Intro}

Lattice field theory, a powerful framework for simulating the behavior of quantum field theories on a discrete lattice, has been an indispensable tool in this pursuit \cite{Wilson:1974sk, Creutz:1980zw}. Among the various challenges faced in lattice field theory, one of the most pervasive and profound is the "doubling problem" associated with lattice fermions \cite{Karsten:1980wd, Nielsen:1980rz, Nielsen:1981xu, Nielsen:1981hk}:
When one discretizes spacetime into a lattice, multiple unwanted fermionic degrees of freedom or "doublers" appear in the theory. These doublers do not correspond to the physical fermions we seek to describe and can lead to unphysical results if not properly addressed.

Over the years, physicists have developed various fermion formulations to tackle this challenge, aiming to recover the desired physics while suppressing the unphysical doublers. These fermion formulations including Wilson fermions \cite{Wilson:1975id}, Domain-wall or overlap fermions \cite{Kaplan:1992bt, Shamir:1993zy, Furman:1994ky, Neuberger:1998wv, Ginsparg:1981bj}, and staggered fermions \cite{Kogut:1974ag, Susskind:1976jm, Kawamoto:1981hw,Sharatchandra:1981si,Golterman:1984cy,Golterman:1985dz,Kilcup:1986dg} represent a crucial area of research in the field, offering insights into the fundamental nature of fermions and guiding our understanding of particle interactions.
The ``non-standard" approaches have been also proposed, including the generalized Wilson fermions \cite{Bietenholz:1999km,Creutz:2010bm,Durr:2010ch,Durr:2012dw, Misumi:2012eh,Cho:2013yha,Cho:2015ffa,Durr:2017wfi},
the staggered-Wilson fermions \cite{Golterman:1984cy, Adams:2009eb,Adams:2010gx,Hoelbling:2010jw, deForcrand:2011ak,Creutz:2011cd,Misumi:2011su,Follana:2011kh,deForcrand:2012bm,Misumi:2012sp,Misumi:2012eh,Durr:2013gp,Hoelbling:2016qfv,Zielinski:2017pko}, the minimally doubled fermion \cite{Karsten:1981gd,Wilczek:1987kw,Creutz:2007af,Borici:2007kz,Bedaque:2008xs,Bedaque:2008jm, Capitani:2009yn,Kimura:2009qe,Kimura:2009di,Creutz:2010cz,Capitani:2010nn,Tiburzi:2010bm,Kamata:2011jn,Misumi:2012uu,Misumi:2012ky,Capitani:2013zta,Capitani:2013iha,Misumi:2013maa,Weber:2013tfa,Weber:2017eds,Durr:2020yqa} 
and the central-branch Wilson fermion \cite{Kimura:2011ik,Creutz:2011cd,Misumi:2012eh,Chowdhury:2013ux}.

In the previous works of ours \cite{Yumoto:2021fkm,Yumoto:2023ums}, we figured out the non-trivial relation between spectral graph theory and lattice field theory, and investigated the number of zero-eigenvalues of lattice Dirac operators in terms of graph theory (See also the literature \cite{Ohta:2021xty,Matsuura2022,Matsuura_graph-zeta_2022,Matsuura2023}).
We proposed a conjecture \cite{Yumoto:2023ums} claiming, ``{\it under certain conditions, the maximal number of Dirac zero-modes of 
a free lattice fermion is equal to the sum of Betti numbers of the graph (lattice) on which the fermion is defined}".
This conjecture is consistent with the known fact on naive lattice fermions: the species of the four-dimensional naive fermion is sixteen, which is interpreted as the sum of the Betti number of four-dimensional torus ($T^4$).

In this study, we investigate operators in lattice field theory using spectral graph theory and present partial evidence supporting the conjecture regarding the interplay between Dirac zero-modes and the Betti numbers of the graph \cite{Yumoto:2023ums}.
Beyond fundamental concepts in graph theory, such as the adjacency matrix and Laplacian matrix, we introduce an "anti-symmetrized adjacency matrix" and explore its rank in relation to graph topology.
It is noteworthy that the graph Laplacian matrix corresponds to a lattice scalar operator and a Wilson term of lattice fermion, while the anti-symmetrized adjacency matrix, along with its higher-dimensional extensions, coincide with a naive Dirac operator for the free lattice theory. Leveraging these equivalences, we elucidate the counts of exact zero-modes for free scalar and free Dirac operators on the lattice, linking them to Betti numbers associated with graph topology:
(i) The number of zero-modes for a free lattice scalar operator is dictated by the zero-th Betti number of the graph (lattice).
(ii) The maximal number of Dirac zero-modes for a free fermion operator is equated to the sum of all Betti numbers of the graph (lattice). This holds true for $D$-dimensional graphs structured as cartesian products of cycle digraphs ($T^1$ lattice) and simple directed paths ($B^1$ lattice).
We also discuss our result indicating that the naive and massless fermion on a certain graph corresponding to a $D$-dimensional sphere has two exact Dirac zero-modes, which is consistent to the fact that the sum of Betti numbers of the $D$-dimensional sphere is two.

This paper is constructed as follows:
In Sec.~\ref{sec:GT} and Sec.~\ref{sec:TP} we review graph theory and matrices defined in the theory. We review the basic theorems and show a novel theorem regarding the anti-symmetrized adjacency matrix.
In Sec.~\ref{sec:LT} we study lattice scalar and Dirac fields in terms of graph theory, and show the theorems on the zero-modes of the scalar and Dirac operators.  
Sec.~\ref{sec:SD} is devoted to the summary and discussion.

\section{Graph Theory and Matrices}
\label{sec:GT}

\subsection{Graph}

We firstly introduce basic notions and definitions of matrices in graph theory.
The definition of a undirected and unweighted graph \cite{west2001introduction,bondy1976graph,mieghem_2010,Watts1998} is given as follows:
\begin{definition}[graph]
	A graph $G$ is a pair $G = (V,E)$, where $V$ is a set of vertices of the graph and $E$ is a set of edges of the graph.
\end{definition}
As examples, we exhibit two graphs in Fig.~\ref{graph:undirected} with $V = \left\{ 1,2,3,4 \right\}$ and $E =  \{\{ 1,2 \}, \{ 1,3 \}, \{ 1,4 \}, \{ 3,4 \} \}$. Here, $\{i,j\}$ stands for an edge from $i$ to $j$.
If every adjacent vertices can be joined by an edge, the graph is referred to as ``connected".
Each of connected pieces of a graph is referred to as a ``connected component".
The two graphs in Fig.~\ref{graph:undirected} are connected, where they have single connected components.
Directed graphs are defined as,
\begin{figure}[]
	\includegraphics[height=4cm]{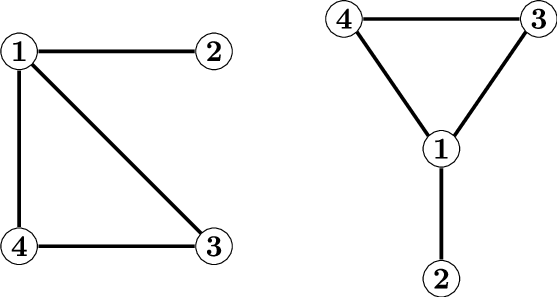}
\vspace{5pt}
\caption{These examples are graphs having a pair $G=(V,E)$ with $V = \left\{ 1,2,3,4 \right\}$ and $E = \left\{ \{ 1,2 \}, \{ 1,3 \}, \{ 1,4 \}, \{ 3,4 \} \right\}$.}
\label{graph:undirected}
\end{figure}

\begin{definition}[directed graph or digraph]
	A directed graph (or digraph) is a pair $(V,E)$ of sets of vertices and edges together with two maps $\init : E \to V$ and $\ter : E \to V$. The two maps are assigned to every edge $e_{ij}$ with an initial vertex $\init (e_{ij}) = v_{i} \in V$ and a terminal vertex $\ter(e_{ij}) = v_{j} \in V$. The edge $e_{ij}$ is said to be directed from $\init (e_{ij})$ to $\ter (e_{ij})$. If $\init(e_{ij}) = \ter(e_{ij})$, the edge $e_{ij}$ is called a loop.
\end{definition}
Two graphs in Fig.~\ref{graph:directed} are digraphs with $V = \left\{ 1,2,3,4 \right\}$ and $E = \left\{ \{ 1,2 \}, \{ 1,3 \}, \{ 1,4 \}, \{ 3,4 \} \right\}$. Initial vertex and terminal vertices are assigned as $\init(\{i,j\}) = i, \ter(\{i,j\}) = j$. Weighted graphs are defined as follows:
\begin{figure}[]
\centering
	\includegraphics[height=4cm]{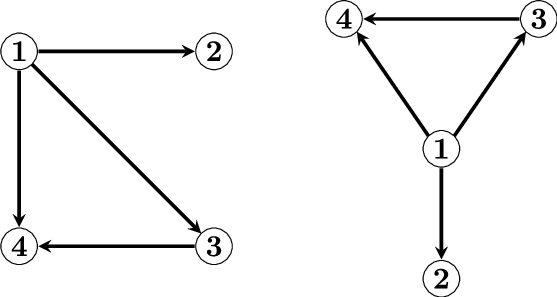}
	\vspace{5pt}
\caption{These examples are graphs having a pair $(V,E)$ with $V = \left\{ 1,2,3,4 \right\}$ and $E = \left\{ \{ 1,2 \}, \{ 1,3 \}, \{ 1,4 \}, \{ 3,4 \} \right\}$. The initial vertices of edges are $\init(\{1,2\}) = 1, \init(\{1,3\}) = 1, \init(\{1,4\}) = 1, \init(\{3,4\}) = 3$, while the terminal vertices are $\init(\{1,2\}) = 2, \init(\{1,3\}) = 3, \init(\{1,4\}) = 4, \init(\{3,4\}) = 4$.}
\label{graph:directed}
\end{figure}

\begin{definition}[weighted graph]
The weighted graph has a value (the weight) for each edge in a graph or a digraph.
\end{definition}
We depict an example of weighted graphs in Fig.~\ref{graph:weighted}. 
It is a weighted and directed graph, each of whose edge has a weight.
\begin{figure}[]
	\includegraphics[height=4cm]{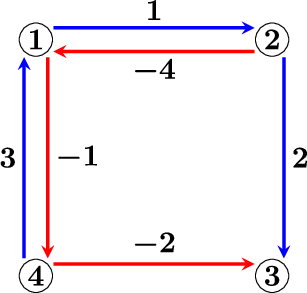}
	\vspace{5pt}
\caption{This digraph is weighted. Blue edges in the graph are those with positive weights, while red edges are those with negative weights.}
\label{graph:weighted}
\end{figure}

\subsection{Matrices}

We here give definitions of matrices associated with graphs. 
We first give a definition of a degree matrix.
\begin{definition}[Degree matrix]
A degree matrix $D$ of a graph is a $|V| \times |V|$ matrix defined as
	\begin{equation}
		D_{ij} = \begin{cases}
			{\rm deg}(v_{i}) 	& \mbox{$i=j$} \\
			0		& \mbox{otherwise}
		\end{cases}\,.
	\end{equation}
\end{definition}
The degree ${\rm deg}(v_{i})$ of a vertex $v_{i}$ counts the number of times an edge terminates at that vertex.
It is defined for both undirected and directed graphs.
As an example we exhibit an degree matrix $D$ of a undirected graph in Fig.~\ref{graph:undirected} as
\begin{equation}
	D = \begin{pmatrix}
		3 & 0 & 0 & 0 \\
		0 & 1 & 0 & 0 \\
		0 & 0& 2 & 0 \\
		0 & 0 & 0 & 2
	\end{pmatrix}\,.
\end{equation}
On the other hand, the degree matrix $D$ of a directed graph in Fig.~\ref{graph:directed} is
\begin{equation}
	D = \begin{pmatrix}
		0 & 0 & 0 & 0 \\
		0 & 1 & 0 & 0 \\
		0 & 0& 1 & 0 \\
		0 & 0 & 0 & 2
	\end{pmatrix}\,.
\end{equation}
\vspace{0.5cm}

We next give a definition of an incidence matrix for undirected matrices.
\begin{definition}[Incidence matrix (undirected)]
An incidence matrix $B$ of a undirected graph is a $|V| \times |E|$ matrix defined as
	\begin{equation}
		B_{ij} = \begin{cases}
			1 	& \mbox{a vertex $v_{i}$ is incident with edge $e_{j}$} \\
			0		& \mbox{otherwise}
		\end{cases}\,.
	\end{equation}
\end{definition}
As an example we exhibit an incidence matrix $B$ of a graph in Fig.~\ref{graph:undirected}
\begin{equation}
	B = \begin{pmatrix}
		1 & 1 & 1 & 0 \\
		1 & 0 & 0 & 0 \\
		0 & 1& 0 & 1 \\
		0 & 0 & 1 & 1
	\end{pmatrix}\,,
\end{equation}
where we define $e_{13} = e_{1}, e_{21} = e_{2}, e_{34} = e_{3}, e_{41} = e_{4}$.

\begin{definition}[Incidence matrix (directed)]
An incidence matrix $B$ of a directed graph is a $|V| \times |E|$ matrix defined as
	\begin{equation}
		B_{ij} = \begin{cases}
		         -1     &  \mbox{an edge $e_{j}$ leaves a vertex $v_{i}$} \\
			1 	& \mbox{an edge $e_{j}$ enters a vertex $v_{i}$} \\
			0		& \mbox{otherwise}
		\end{cases}\,.
	\end{equation}
\end{definition}
The incidence matrix $B$ of a graph in Fig.~\ref{graph:directed} is
\begin{equation}
	B = \begin{pmatrix}
		-1 & -1 & -1 & 0 \\
		1 & 0 & 0 & 0 \\
		0 & 1 & 0 & -1 \\
		0 & 0 & 1 & 1
	\end{pmatrix}\,.
\end{equation}
where we again define $e_{13} = e_{1}, e_{21} = e_{2}, e_{34} = e_{3}, e_{41} = e_{4}$.
\vspace{0.5cm}

We next give a definition of an adjacency matrix.
An adjacency matrix for unweighted graphs is defined as follows:
\begin{definition}[adjacency matrix (unweighted)]
\label{def:adj}
	The adjacency matrix $A$ of a graph is the $|V| \times |V|$ matrix given by
	\begin{equation}
		A_{ij} = \begin{cases}
			1	& \mbox{if there is an edge between $i$ and $j$} \\
			0		& \mbox{otherwise}
		\end{cases}\,.
	\end{equation}
\end{definition}
This matrix is symmetric by definition.
As an example, we exhibit an adjacency matrix $A$ of a graph in Fig.~\ref{graph:undirected} as
\begin{equation}
	A = \begin{pmatrix}
		0 & 1 & 1 & 1 \\
		1 & 0 & 0 & 0 \\
		1 & 0& 0 & 1 \\
		1 & 0 & 1 & 0
	\end{pmatrix}\,.
\end{equation}
Furthermore, an adjacency matrix for weighted graphs is defined as follows.
\begin{definition}[adjacency matrix (weighted)]
	The adjacency matrix $A$ of a graph is the $|V| \times |V|$ matrix given by
	\begin{equation}
		A_{ij} = \begin{cases}
			w_{ij}	& \mbox{if there is a edge from $i$ to $j$} \\
			0		& \mbox{otherwise}
		\end{cases}\,,
	\end{equation}
	where $w_{ij}$ is the weight of an edge from $i$ to $j$.
\end{definition}
As an example we exhibit an adjacency matrix $A$ of a graph in Fig.~\ref{graph:weighted}
\begin{equation}
	A = \begin{pmatrix}
		0 & 1 & 0 & -1 \\
		-4 & 0 & 2 & 0 \\
		0 & 0 & 0 & 0 \\
		3 & 0 & -2 & 0
	\end{pmatrix}\,.
\end{equation}
In general, the adjacency matrix of a directed graph is asymmetric since the existence of an edge from $i$ to $j$ does not necessarily imply that there is also an edge from $j$ to $i$.
\vspace{0.5cm}

The Laplacian matrix is defined by use of the degree, adjacency and incidence matrices as follows:
\begin{definition}[Laplacian matrix]
The Laplacian matrix $L$ of a graph is the $|V| \times |V|$ matrix given by
	\begin{equation}
		L = D - A = BB^{T} 
  \label{eq:Lap0}
	\end{equation}
$D,A$ are degree, adjacency matrices of a undirected and unweighted graph,
while $B$ is an incidence matrix of directed graph.
\end{definition}
For a graph in Fig.~\ref{graph:undirected},
the Laplacian matrix is
\begin{equation}
	L = \begin{pmatrix}
		3 & -1 & -1 & -1 \\
		-1 & 1 & 0 & 0 \\
		-1 & 0& 2 & -1 \\
		-1 & 0 & -1 & 2
	\end{pmatrix}\,.
\end{equation}
One can easily check out $L = D - A = BB^{T} $.
\vspace{0.5cm}

We extend the notion of graph Laplacian matrix to higher-dimensional graphs, where a face $F$ can be defined on an area surrounded by a loop (or a cycle) of edges.  
For such a generalized directed graph, the combinatorial Laplacian matrix is defined as follows:
(We note that the symbols $l$ or $k$ stand for edges. )
\begin{definition}[combinatorial Laplacian matrix]
The combinatorial Laplacian matrix $L_{1}$ of a graph is the $|E| \times |E|$ matrix given by
	\begin{equation}
		(L_{1})_{lk} = \begin{cases}
			2+ {\rm deg}(F_l) 	& \mbox{l = k} \\
			1		&  \mbox{if both l and k go in or out of the same vertex  } \\
                -1    &  \mbox{if one goes in and the other goes out of the same vertex  } \\
                 0   & \mbox{otherwise} 
		\end{cases}\,,
   \label{eq:Lap1}
	\end{equation}
where ${\rm deg}(F_{l})$ is the number of faces touching the edge $l$.
\end{definition}
As an example, we take a graph in Fig.~\ref{graph:directed}.
There are two choices: the area enclosed by the loop of edges can be regarded as a face or a hole.
We here name four edges as $12,13,34,14$.
The combinatorial Laplacian matrix for the graph with a face is
\begin{equation}
	L_1 = \begin{pmatrix}
		2 & 1 & 0 & 1 \\
		1 & 3 & -1 & 1 \\
		0 & -1& 3 & 1 \\
		1 & 1 & 1 & 3
	\end{pmatrix}\,.
 \label{eq:nlcm}
\end{equation}
The combinatorial Laplacian matrix for the graph without a face (with a hole) is
\begin{equation}
	L_1 = \begin{pmatrix}
		2 & 1 & 0 & 1 \\
		1 & 2 & -1 & 1 \\
		0 & -1& 2 & 1 \\
		1 & 1 & 1 & 2
	\end{pmatrix}\,.
  \label{eq:lcm}
\end{equation}
\vspace{0.5cm}

In addition to the standard definitions of graph matrices, we introduce a specific adjacency matrix of directed graph as an ``anti-symmetrized adjacency matrix" as follows.
\begin{definition}[Anti-symmetrized adjacency matrix]
\label{def:ASAM}
	The anti-symmetrized adjacency matrix $A_{\rm as}$ of a directed graph having no multiple edges is the $|V| \times |V|$ matrix given by
	\begin{equation}
		(A_{\rm as})_{ij} = \begin{cases}
			1	& \mbox{an edge leaves $i$ and enters $j$} \\
			-1	& \mbox{an edge leaves $j$ and enters $i$} \\
			0		& \mbox{otherwise}
		\end{cases}\,,
	\end{equation}
\end{definition}
The anti-symmetrized adjacency matrix of a directed graph is anti-symmetric since it is expressed as $A_{\rm as} = A - A^{T}$, where $A$ is the adjacency matrix defined in Def.~\ref{def:adj}.
The anti-symmetrized adjacency matrix $A_{\rm as}$ of a graph in Fig.~\ref{graph:directed} is
\begin{equation}
	A_{\rm as} = \begin{pmatrix}
		0 & 1 & 1 & 1 \\
		-1 & 0 & 0 & 0 \\
		-1 & 0 & 0 & 1 \\
		-1 & 0 & -1 & 0
	\end{pmatrix}\,.
\end{equation}

\section{Topology of graph}
\label{sec:TP}

\subsection{Betti numbers and Laplacian}

The topology of a graph is known to be detected by matrices associated with the graph.
\begin{theorem}[Betti numbers and Laplacian]
The zeroth and first Betti numbers $\beta_{0},\beta_{1}$ of the graphs are related to the rank of graph Laplacian.
	\begin{align}
    \label{eq:lap_betti0}
          |V| - {\rm rank}(L) &= \beta_{0}\,,
          \\
    \label{eq:lap_betti1}
          |E| - {\rm rank}(L) &= \beta_{1}\,.
	\end{align}
\end{theorem} 
It means that the number of exact zero eigenvalues of the Laplacian matrix agrees with $\beta_{0}$, which counts simply connected parts of the graph. The proof of this theorem is given in Appendix.~\ref{app:Betti_laplacian}.
In the case of the connected graph, the claim of this theorem have been shown using the incidence matrices \cite{Ohta:2021xty}.

The topology of a graph is also detected by the combinatorial Laplacian matrix.
\begin{theorem}[Betti numbers and combinatorial Laplacian]
The first Betti numbers $\beta_{1}$ of the graphs are related to the rank of the combinatorial Laplacian matrix.
	\begin{align}
          |E| - {\rm rank}(L_{1}) &= \beta_{1}\,.
	\end{align}
\end{theorem} 
It means that the number of exact zero eigenvalues of the combinatorial Laplacian matrix agrees with $\beta_{1}$, which counts the number of holes in the graph.
The combinatorial Laplacian in Eq.~(\ref{eq:nlcm}) gives $\beta_{1}=0$ since there is no hole, while The one in Eq.~(\ref{eq:lcm}) gives $\beta_{1}=1$ since there is a hole.

\subsection{Anti-symmetrized adjacency matrix}

In this subsection, we prove a novel theorem associated with the anti-symmetrized adjacency matrix

\subsubsection{Simple cases in 1D}

We here consider a directed graph corresponding to one-dimensional lattices, where we take a specific orientation as shown in Fig.~\ref{digraph:cyc_simpath}. 
We take two types of the graphs, a cycle digraph (left in Fig.~\ref{digraph:cyc_simpath}) and a simple directed path (right in Fig.~\ref{digraph:cyc_simpath}).
These two graphs are identified as a one-dimensional torus lattice $T^1$ (or circle lattice $S^1$) and a one-dimensional disk lattice (or hyperball lattice) $B^1$, respectively.
For these graphs or cartesian products of these graph elements, we have the following theorems:
\begin{figure}[htpb]
\begin{minipage}[t]{0.45\linewidth}
\centering
	\includegraphics[clip,
		height=4cm
		]{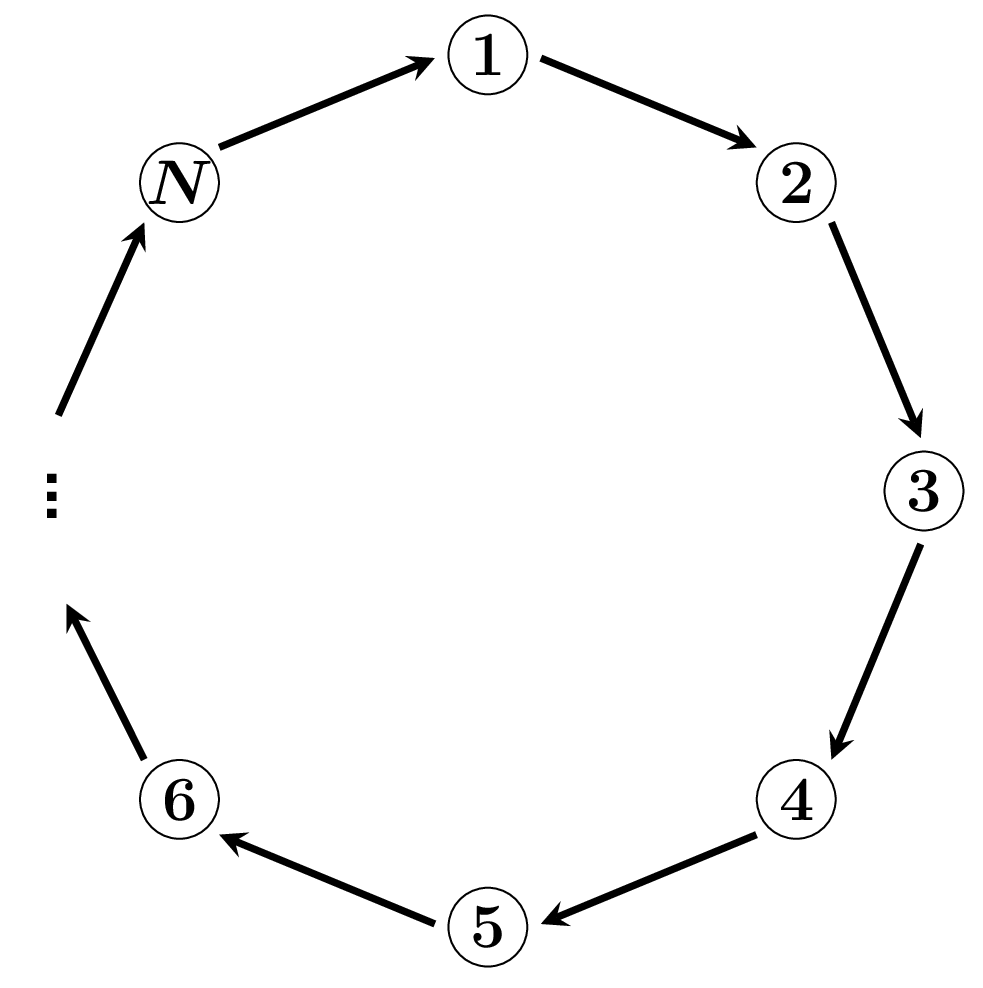}
\subcaption{A cycle digraph corresponds to 1D lattice with PBC.}
\label{digraph:cycle}
\end{minipage}
\begin{minipage}[t]{0.45\linewidth}
\centering
	\includegraphics[clip,
		height=4cm
		]{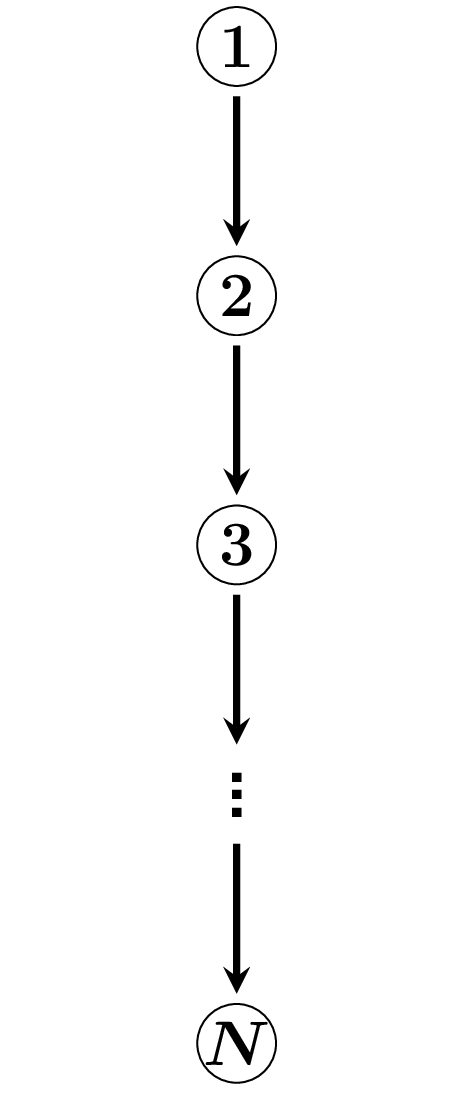}
\subcaption{A simple directed path corresponds to 1D lattice with DBC.}
\label{digraph:simple_path}
\end{minipage}
\caption{
The two graphs correspond to a cycle digraph $D_\mathrm{c}$ and a simple directed path $D_\mathrm{p}$ respectively.
}
\label{digraph:cyc_simpath}
\end{figure}

\begin{theorem}[]
\label{thm3}
For the two types of the graphs $G$ in Fig.~\ref{digraph:cyc_simpath}, which are a cycle digraph and a simple directed path, a following inequality holds;
\begin{align}
\label{ineq:thm3}
    |V| - \rank A_{\rm as}(G) \geq   \beta_{1}(G)\,.
\end{align}
This inequality can be rewritten as
\begin{equation}
\label{ineq:thm3_v2}
    \dim \left( \ker A_{\rm as}(G) \right) \geq 1 - |V| + |E|\,,
\end{equation}
since the first Betti number of the graph is $\beta_{1}(G) = 1 - |V| + |E|$.
\end{theorem}

\begin{proof}
To prove this theorem, we will show that the inequality in Eq.~(\ref{ineq:thm3}) holds for any number of vertices in the cycle digraph or the simple directed path.
Firstly, we prove it for the cycle digraph, which we denote as $D_\mathrm{c}$.
Note that the minimal number of vertices in the cycle digraph is $3$ (triangle graph).
The right side in the inequality Eqs.~(\ref{ineq:thm3})(\ref{ineq:thm3_v2}) is equal to one ($\beta_1 (D_\mathrm{c}) = 1$) since the number of edges is equal to the number of vertices for the cycle digraph ($|E|=|V|$).
An anti-symmetrized adjacency matrix of the cycle digraph $D_\mathrm{c}$ is given as
\begin{equation}
    A_\mathrm{as}(D_\mathrm{c}) = 
    \begin{pmatrix}
        0	&1		&0		&		&0		&0		&-1 \\
		-1	&0		&1		&\cdots	&0		&0		&0 \\
		0	&-1		&0		&		&0		&0		&0 \\
			&\vdots	&		&\ddots	&		&\vdots	& \\
		0	&0		&0		&		&0		&1		&0 \\
		0	&0		&0		&\cdots	&-1		&0		&1 \\
		1	&0		&0		&		&0		&-1		&0
    \end{pmatrix}\,,
\end{equation}
from Def.~\ref{def:ASAM}.
The diagonalization of this matrix can be derived as 
\begin{equation}
\label{eq:cyc_diag}
    \left(U^\dagger A_\mathrm{as}(D_\mathrm{c}) U\right)_{mn}
    = 2i\sin \left( \frac{2\pi(m-1)}{|V|} \right) \delta_{mn}\,,
\end{equation}
where $m$ and $n$ are integers and run from $1$ to $|V|$.
$U$ is the unitary matrix for diagonalization and $\delta_{mn}$ is the Kronecker delta.
When the anti-symmetirized adjacency matrix has zero-modes, they are given by solutions of the equation
\begin{equation}
    \sin \left( \frac{2\pi(m-1)}{|V|} \right) = 0\,,
\end{equation}
since zero-modes are defined by $\left(U^\dagger A_\mathrm{as}(D_\mathrm{c}) U\right)_{mn} = 0$.
The solutions of this equation are $m = 1$ or $m = \frac{|V|}{2} + 1$.
Note that the solution becomes only $m = 1$ when the number of vertices is odd because $\frac{|V|}{2} + 1$ is not an integer for odd $|V|$.
Therefore, when the number of vertices is odd, $A_\mathrm{as}(D_\mathrm{c})$ has only one zero-mode.
By contrast, $A_\mathrm{as}(D_\mathrm{c})$ has two zero-modes when the number of vertices is even.
We can summarize these facts as follows:
\begin{itemize}
    \item 
    The number of vertices is odd; $|V| - \rank A_{\rm as}(D_\mathrm{c}) = 1$,
    \item 
    The number of vertices is even; $|V| - \rank A_{\rm as}(D_\mathrm{c}) = 2$.
\end{itemize}
Hence, the inequality $|V| - \rank A_{\rm as}(D_\mathrm{c}) \geq \beta_{1}(D_\mathrm{c}) = 1$ holds for the cycle digraph. 

Next, we prove it for the simple directed path, which we denote as $D_\mathrm{p}$.
The minimal number of vertices of this digraph is $2$.
The right side in the inequality Eqs.~(\ref{ineq:thm3})(\ref{ineq:thm3_v2}) is equal to zero ($\beta_1 (D_\mathrm{p}) = 0$) since the number of edges is one less than the number of vertices for the simple directed path $D_\mathrm{p}$.
An anti-symmetrized adjacency matrix of $D_\mathrm{p}$ is 
\begin{equation}
    A_\mathrm{as}(D_\mathrm{p}) = 
    \begin{pmatrix}
        0	&1		&0		&		&0		&0		&0 \\
		-1	&0		&1		&\cdots	&0		&0		&0 \\
		0	&-1		&0		&		&0		&0		&0 \\
			&\vdots	&		&\ddots	&		&\vdots	& \\
		0	&0		&0		&		&0		&1		&0 \\
		0	&0		&0		&\cdots	&-1		&0		&1 \\
		0	&0		&0		&		&0		&-1		&0
    \end{pmatrix}.
\end{equation}
The diagonalization of this matrix can be derived as 
\begin{equation}
\label{eq:simpath_diag}
    \left(V^\dagger A_\mathrm{as}(D_\mathrm{p}) V\right)_{mn}
    = 2i\cos \left( \frac{m\pi}{|V|+1} \right) \delta_{mn}\,,
\end{equation}
where $m$ and $n$ are integers and run from $1$ to $|V|$.
And $V$ is the unitary matrix different from $U$.
When this anti-symmetirized adjacency matrix has zero-modes, they are solutions of the following equation,
\begin{equation}
    \cos \left( \frac{m\pi}{|V|+1} \right) = 0\,.
\end{equation}
The solution of this equation is $m = \frac{|V|+1}{2}$.
It means that the anti-symmetrized adjacency matrix $A_\mathrm{as}(D_\mathrm{p})$ has only one zero-mode for the number of vertices being odd while $A_\mathrm{as}(D_\mathrm{p})$ has no zero-mode for the number of vertices being even.
We can summarize these facts as follows;
\begin{itemize}
    \item 
    The number of vertices is odd; $|V| - \rank A_\mathrm{as}(D_\mathrm{p}) = 1$,
    \item 
    The number of vertices is even; $|V| - \rank A_\mathrm{as}(D_\mathrm{p}) = 0$.
\end{itemize}
Hence, the inequality $|V| - \rank A_{\rm as}(D_\mathrm{p}) \geq \beta_{1}(D_\mathrm{p})= 0$ holds for the simple directed path. 
\end{proof}

\subsubsection{Special cases in 1D}

If we restrict our graphs in the previous subsection to $|V| = {\rm even}$ for $D_\mathrm{c}$ (torus) and $|V|= {\rm odd}$ for $D_\mathrm{c}$ (disk),
the theorem is modified as follows:
\begin{theorem}[]
\label{thm4}
When the graphs $G$ are the cycle digraph with even vertices or the simple directed path with odd vertices, the following equation holds;
\begin{align}
\label{eq:thm4}
    |V| - \rank A_{\rm as}(G) =  \beta_{0}(G) +  \beta_{1}(G)\,.
\end{align}
This equation can be rewritten as
\begin{equation}
\label{eq:thm4_v2}
    \dim \left(\ker A_{\rm as}(G) \right) = 2 - |V| + |E|\,.
\end{equation}
where $\beta_{0}(G)=1$.
\end{theorem} 

\begin{proof}
    Firstly, we prove it for the cycle digraph $D_\mathrm{c}$ with even vertices.
    From the proof in Thm.~\ref{thm3}, the rank of anti-symmetrized adjacency matrix of the cycle digraph is $|V|-2$ when the number of vertices is even.
    On the other hand, we have $\beta_{0}(D_\mathrm{c}) + \beta_{1}(D_\mathrm{c}) =2 - |V| +|E| = 2$ since the number of edges is equal to the number of vertices.
    Hence, the equation $|V| - \rank A_{\rm as}(D_\mathrm{c}) = 2 =\beta_{0}(D_\mathrm{c}) + \beta_{1}(D_\mathrm{c})$ holds in the cycle digraph with even vertices.
    
    Secondly, we prove the theorem for the simple directed path $D_\mathrm{p}$ with odd vertices.
    From the proof in Thm.~\ref{thm3}, the rank of anti-symmetrized adjacency matrix of the simple directed path is $|V|-1$ when the number of vertices is odd.
    On the other hand, we have $\beta_{0}(D_\mathrm{p}) + \beta_{1}(D_\mathrm{c}) =2 - |V| +|E|= 1$ since the number of edges is one less than the number of vertices.
    Hence, the equation $|V| - \rank A_{\rm as}(D_\mathrm{c}) = 1 = \beta_{0}(D_\mathrm{c}) + \beta_{1}(D_\mathrm{c})$ holds for the simple directed path with odd vertices.
\end{proof}

If we consider the cycle digraph or the simple directed path with an arbitrary number of vertices, the following inequality holds;
\begin{equation}
    \dim \left(\ker A_{\rm as}(G) \right) \leq \beta_{0}(G) +  \beta_{1}(G) = 2 - |V| + |E|\,.
\end{equation}


\subsubsection{Higher dimensions}

We consider graphs $G$ defined as a cartesian-product $\Box$ of only the cycle digraph $D_\mathrm{c}$ and the simple directed path $D_\mathrm{p}$, that is
\begin{equation}
\label{eq:D-dim_graph}
    G \equiv G_1 \cprod G_2 \cprod \cdots \cprod G_D \,, 
\end{equation}
where $G_\mu \in \left\{ D_\mathrm{c}, D_\mathrm{p} \right\}$ for $\mu \in \left\{1,2,\cdots,D\right\}$.
In this paper, we term the graphs in Eq.~(\ref{eq:D-dim_graph}) as ``$D$-dimensional directed graphs".
As some examples, $G = D_\mathrm{c} \cprod D_\mathrm{c}$, $G = D_\mathrm{p} \cprod D_\mathrm{p}$, and $G = D_\mathrm{c} \cprod D_\mathrm{p}$ are depicted in Fig.~\ref{graph:examples}.
\begin{figure}[htpb]
\begin{minipage}[t]{\linewidth}
\centering
	\includegraphics[
    scale=.4
    ]{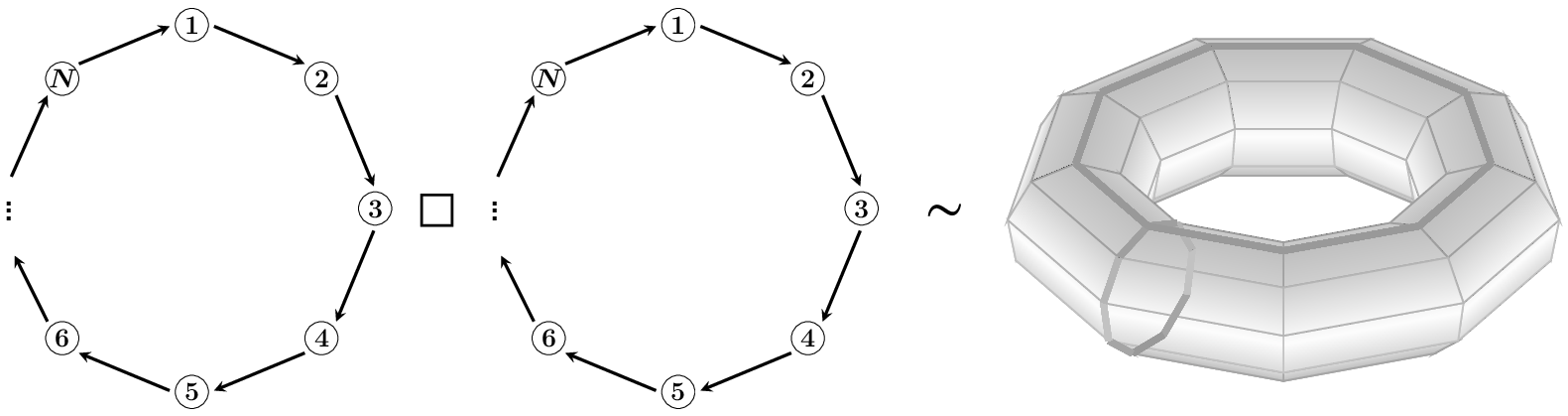}
\subcaption{$G = D_\mathrm{c} \Box D_\mathrm{c}$ is identified as a two dimensional torus with directed edge.}
\label{graph:cc}
\end{minipage}\vspace{5pt}\\
\begin{minipage}[t]{\linewidth}
\centering
	\includegraphics[
    scale=.4
    ]{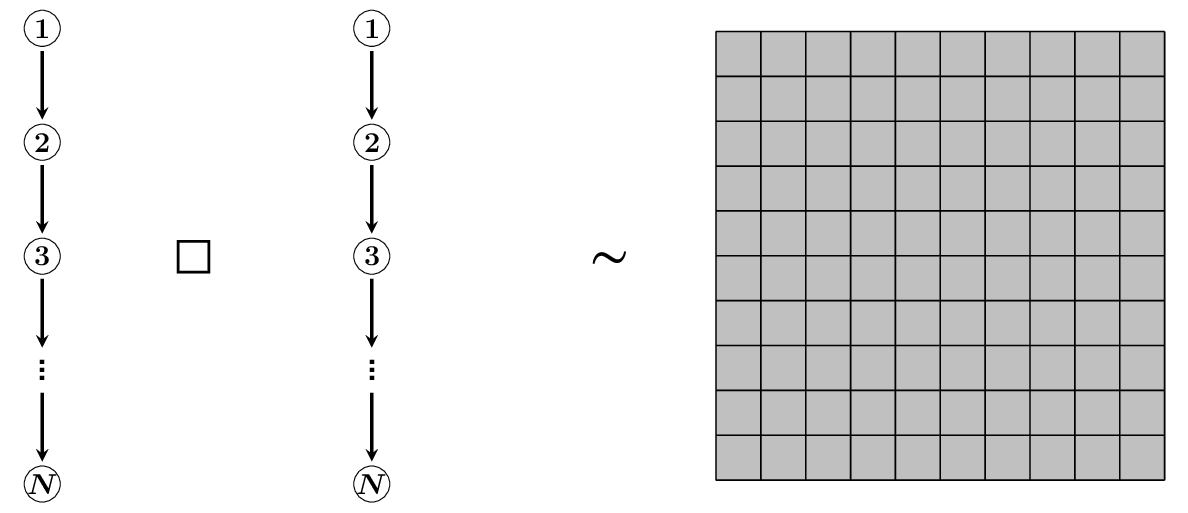}
\subcaption{$G = D_\mathrm{p} \Box D_\mathrm{p}$ is identified as a two dimensional disk with directed edge.}
\label{graph:pp}
\end{minipage}\vspace{5pt}\\
\begin{minipage}[t]{\linewidth}
\centering
	\includegraphics[
    scale=.4
    ]{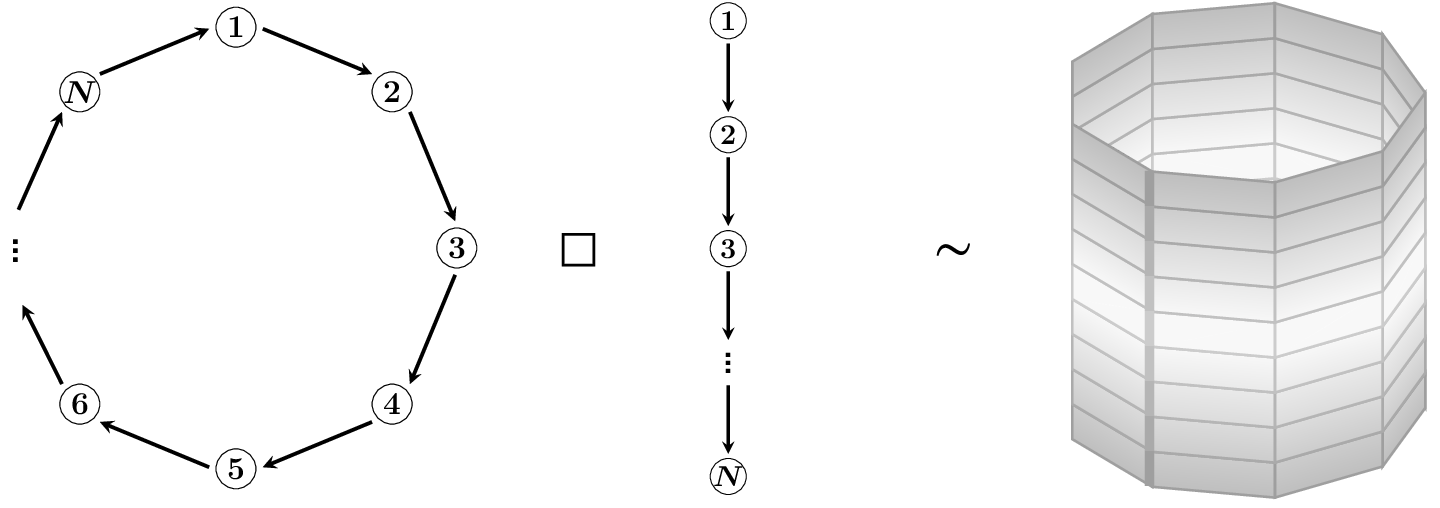}
\subcaption{$G = D_\mathrm{c} \Box D_\mathrm{p}$ is identified as an cylinder with directed edge.}
\label{graph:pp}
\end{minipage}
\caption{
Three examples for $G$ in Eq.~(\ref{eq:D-dim_graph}) and manifolds corresponding to them are depicted.
}
\label{graph:examples}
\end{figure}
For such directed graphs, we define the following matrix.
\begin{equation}
\label{eq:D_matrix}
\begin{split}
    \mathcal{D}(G) &\equiv \sum_{\mu=1}^{D} 
    \left\{\left( \bigot_{\nu=1}^{D-\mu} \bm{1}_{|V|_{D+1-\nu}} \right)
    \otimes A_\mathrm{as} (G_\mu) \otimes
    \left( \bigot_{\rho=1}^{\mu-1} \bm{1}_{|V|_\rho} \right)\right\} \ot \gamma_\mu\\
    &= \sum_{\mu=1}^{D} \left(A_\mathrm{as}\right)_\mu \ot \gamma_\mu \,,
\end{split}
\end{equation}
where $|V|_\mu$ stands for the number of vertices in $G_\mu$ and $\bm{1}_{|V|_\mu}$ is the identity matrix of size $|V|_\mu$.
And $\gamma_{\mu}$ is $D$-dimensional gamma matrices satisfying Clifford algebra as 
$\{\gamma_{\mu},\gamma_{\nu} \} = 2\delta_{\mu\nu}$.
The symbol $\ot$ is Kronecker product.

Based on Thm.~\ref{thm3}, we claim the following theorem on the rank of $\mathcal{D}(G)$ for graphs $G$ satisfying $|V| = {\rm even}$ for the cycle digraph and $|V|= {\rm odd}$ for the simple directed path.
\begin{theorem}[]
\label{thm5}
For the graphs $G$ constructed a cartesian-product of only the cycle digraph with even vertices and the simple directed path with odd vertices, the following equation holds;
\begin{equation}
\label{eq:thm5}
    |V|\cdot \rank \gamma - \rank \mathcal{D}(G) = \rank \gamma \cdot \prod_{\mu=1}^{D} \Bigl\{ \beta_0(G_\mu) + \beta_1(G_\mu) \Bigr\}\,,
\end{equation}
where $|V|$ is the number of vertices in $G$ and $\rank \gamma$ is the rank of gamma matrices.
When the graph $G$ is constructed by $d$ cycle digraphs and $(D-d)$ simple directed paths, this equation is rewritten as
\begin{equation}
\label{eq:thm5_v2}
    \frac{\dim \left(\ker \mathcal{D}(G) \right)}{\rank \gamma} = 2^{d}\,.
\end{equation}
\end{theorem}

\begin{proof}
We consider the graph $G$ constructed as a cartesian-product of $d$ cycle digraphs and $(D-d)$ simple directed paths.
We divide $\mu \in \left\{1,2,\cdots,D\right\}$ in Eq.~(\ref{eq:D-dim_graph}) into two parts as $S_\mathrm{c} \equiv \left\{ \mu \setmid G_\mu = D_\mathrm{c} \right\}$ and $S_\mathrm{p} \equiv \left\{ \mu \setmid G_\mu = D_\mathrm{p} \right\}$ in order to distinguish the cycle digraph $D_\mathrm{c}$ and the simple directed path $D_\mathrm{p}$.
The number of dimensions of $S_\mathrm{c}$ and $S_\mathrm{p}$, which are denoted as $|S_\mathrm{c}|$ and $|S_\mathrm{p}|$, are $|S_\mathrm{c}| = d$ and $|S_\mathrm{p}| = D-d$, respectively.
Then, the right side in Eq.~(\ref{eq:thm5}) is
\begin{equation}
\begin{split}
    &\rank \gamma \cdot \prod_{\mu=1}^{D} \Bigl\{ \beta_0(G_\mu) + \beta_1(G_\mu) \Bigr\}\\
    &\quad= \rank \gamma \cdot \prod_{\mu_\mathrm{c} \in S_\mathrm{c}} 
        \Bigl\{ \beta_0(D_\mathrm{c}) + \beta_1(D_\mathrm{c}) \Bigr\}
    \prod_{\mu_\mathrm{p} \in S_\mathrm{p}} 
        \Bigl\{ \beta_0(D_\mathrm{p}) + \beta_1(D_\mathrm{p}) \Bigr\}\\
    &\quad= \rank \gamma \cdot 
        \prod_{\mu_\mathrm{c} \in S_\mathrm{c}} 2 \cdot 
        \prod_{\mu_\mathrm{p} \in S_\mathrm{p}} 1\\
    &\quad= \rank \gamma \cdot 2^{d}\,,
\end{split}
\end{equation}
where $\beta_0(D_\mathrm{c}) + \beta_1(D_\mathrm{c}) = 2$ and $\beta_0(D_\mathrm{p}) + \beta_1(D_\mathrm{p}) = 1$.
To prove the theorem in Eq.~(\ref{eq:thm5}), what we have to do is to show that the the left side of Eq.~(\ref{eq:thm5}) satisfies $|V|\cdot \rank \gamma - \rank \mathcal{D}(G) = \rank \gamma \cdot 2^{d}$.

We will derive the number of zero-modes of $\mathcal{D}(G)$.
The diagonalization of $\mathcal{D}(G)$ can be derived as
\begin{equation}
    \mathcal{U}^\dagger \mathcal{D}(G)\, \mathcal{U} = \sum_{\mu=1}^{D} 
    \left\{\left( \bigot_{\nu=1}^{D-\mu} \bm{1}_{|V|_{D+1-\nu}} \right)
    \otimes \left( U_\mu^\dagger A_\mathrm{as} (G_\mu) U_\mu \right) \otimes
    \left( \bigot_{\rho=1}^{\mu-1} \bm{1}_{|V|_\rho} \right)\right\} \ot \gamma_\mu\,,\\
\end{equation}
where $\mathcal{U}$ is an unitary matrix defined as $\mathcal{U} \equiv \bigot_{\mu=1}^{D} U_\mu$ and $U_\mu$ is the unitary matrix for the diagonalization of $A_\mathrm{as}(G_\mu)$.
Consequently, the eigenvalues of $\mathcal{D}(G)$ are obtained as
\begin{equation}
\label{eq:diag_D}
\begin{split}
    &\left( \mathcal{U}^\dagger \mathcal{D}(G)\, \mathcal{U} \right)_{mn} \\
    &\quad= \sum_{\mu=1}^{D} \gamma_\mu
    \left\{ \prod_{\nu=1}^{D-\mu} \prod_{\rho=1}^{\mu-1}
        \left( \bm{1}_{|V|_{D+1-\nu}} \right)_{m_{D+1-\nu} n_{D+1-\nu}}
        \left( \bm{1}_{|V|_\rho} \right)_{m_\rho n_\rho}
    \left( U_\mu^\dagger A_\mathrm{as} (G_\mu) U_\mu \right)_{m_\mu n_\mu} \right\}\\
    &\quad= \sum_{\mu=1}^{D} \gamma_\mu
        \left( U_\mu^\dagger A_\mathrm{as} (G_\mu) U_\mu \right)_{m_\mu n_\mu}
        \left(\prod_{\nu=1}^{D-\mu} \prod_{\rho=1}^{\mu-1}
            \delta_{m_{D+1-\nu} n_{D+1-\nu}} \delta_{m_\rho n_\rho} \right)\\
    &\quad= \sum_{\mu_\mathrm{c} \in S_\mathrm{c}} \gamma_{\mu_\mathrm{c}}
        \left( 
            U_{\mu_\mathrm{c}}^\dagger A_\mathrm{as} (D_\mathrm{c}) U_{\mu_\mathrm{c}} 
        \right)_{m_{\mu_\mathrm{c}} n_{\mu_\mathrm{c}}}
        \left(\prod_{\nu=1}^{D-\mu_\mathrm{c}} \prod_{\rho=1}^{\mu_\mathrm{c}-1}
            \delta_{m_{D+1-\nu} n_{D+1-\nu}} \delta_{m_\rho n_\rho} \right)\\
    &\qquad+ \sum_{\mu_\mathrm{p} \in S_\mathrm{p}} \gamma_{\mu_\mathrm{p}}
        \left( 
            U_{\mu_\mathrm{p}}^\dagger A_\mathrm{as} (D_\mathrm{p}) U_{\mu_\mathrm{p}} 
        \right)_{m_{\mu_\mathrm{p}} n_{\mu_\mathrm{p}}}
        \left(\prod_{\nu=1}^{D-\mu_\mathrm{p}} \prod_{\rho=1}^{\mu_\mathrm{p}-1}
            \delta_{m_{D+1-\nu} n_{D+1-\nu}} \delta_{m_\rho n_\rho} \right)\,,
\end{split}
\end{equation}
where $n_\alpha$ and $m_\alpha$ range from 1 to $|V|_\alpha$.
And $m$ and $n$ are
\begin{equation}
\begin{split}
    m 
    = m_1 + \sum_{\mu=2}^{D} \left( \prod_{\nu=1}^{\mu-1}|V|_{\nu} \left( m_\mu -1 \right) \right)\,,\quad
    n 
    = n_1 + \sum_{\mu=2}^{D} \left( \prod_{\nu=1}^{\mu-1}|V|_\nu \left( n_\mu -1 \right) \right)\,,
\end{split}
\end{equation}
respectively.
Here we used $\left( \bm{1}_{|V|_\alpha} \right)_{m_\alpha n_\alpha} = \delta_{m_\alpha n_\alpha}$ and the properties of Kronecker product, which is $\left( A \ot B \right)_{Q(i-1)+k,Q(j-1)+l} = A_{ij}B_{kl}$ for $P$-square matrix $A$ and $Q$-square matrix $B$.
Since the components of $U_\mu^\dagger A_\mathrm{as} (G_\mu) U_\mu$ are
\begin{equation}
    \left( U_\mu^\dagger A_\mathrm{as} (G_\mu) U_\mu \right)_{m_\mu n_\mu} = 
    \begin{cases}
        2i\sin \left( \frac{2\pi(m_\mu-1)}{|V|_\mu} \right) \delta_{m_\mu n_\mu}
            &\left( G_\mu=D_\mathrm{c} \right)\\
        2i\cos \left( \frac{m_\mu \pi}{|V|_\mu +1} \right) \delta_{m_\mu n_\mu}
            &\left( G_\mu=D_\mathrm{p} \right)
    \end{cases}\,.
\end{equation}
From the proof of Thm.~\ref{thm3}, the equation in Eq.~(\ref{eq:diag_D}) is 
\begin{equation}
\begin{split}
    &\left( \mathcal{U}^\dagger \mathcal{D}(G)\, \mathcal{U} \right)_{mn}\\
    &\quad= \sum_{\mu_\mathrm{c} \in S_\mathrm{c}} 2i \gamma_{\mu_\mathrm{c}} 
        \sin \left( \frac{2\pi(m_{\mu_\mathrm{c}}-1)}{|V|_{\mu_\mathrm{c}}} \right) 
        \left(\prod_{\nu=1}^{D-\mu_\mathrm{c}} \prod_{\rho=1}^{\mu_\mathrm{c}-1}
            \delta_{m_{D+1-\nu} n_{D+1-\nu}} \delta_{m_{\mu_\mathrm{c}} n_{\mu_\mathrm{c}}} \delta_{m_\rho n_\rho} \right)\\
    &\qquad+ \sum_{\mu_\mathrm{p} \in S_\mathrm{p}} 2i \gamma_{\mu_\mathrm{p}} 
        \cos \left( \frac{m_{\mu_\mathrm{p}} \pi}{|V|_{\mu_\mathrm{p}} +1} \right) 
        \left(\prod_{\nu=1}^{D-\mu_\mathrm{p}} \prod_{\rho=1}^{\mu_\mathrm{p}-1}
            \delta_{m_{D+1-\nu} n_{D+1-\nu}} \delta_{m_{\mu_\mathrm{p}} n_{\mu_\mathrm{p}}} \delta_{m_\rho n_\rho} \right)\\
    &\quad= 2i \left\{
        \sum_{\mu_\mathrm{c} \in S_\mathrm{c}} \gamma_{\mu_\mathrm{c}}
        \sin \left( \frac{2\pi(m_{\mu_\mathrm{c}}-1)}{|V|_{\mu_\mathrm{c}}} \right)
        +
        \sum_{\mu_\mathrm{p} \in S_\mathrm{p}} \gamma_{\mu_\mathrm{p}} 
        \cos \left( \frac{m_{\mu_\mathrm{p}} \pi}{|V|_{\mu_\mathrm{p}} +1} \right) 
    \right\} \delta_{mn}\,,
\end{split}
\end{equation}
where $\delta_{mn} = \prod_{\mu=1}^{D}\delta_{m_\mu n_\mu}$.
However, since $\gamma$-matrices are linearly independent, the coefficient of each $\gamma$-matrices must be zero.
Finally, the conditions for the diagonal matrix to have zero-modes are below
\begin{equation}
    \sin \left( \frac{2\pi(m_{\mu_\mathrm{c}}-1)}{|V|_{\mu_\mathrm{c}}} \right)
    = \cos \left( \frac{m_{\mu_\mathrm{p}} \pi}{|V|_{\mu_\mathrm{p}}+1} \right) = 0\,.
\end{equation}
The solutions of this equation are $m_{\mu_\mathrm{c}} = 1,\, m_{\mu_\mathrm{p}} = \frac{|V|_{\mu_\mathrm{p}}+1}{2}$ or $m_{\mu_\mathrm{c}} = \frac{|V|_{\mu_\mathrm{c}}}{2} + 1,\ m_{\mu_\mathrm{p}} = \frac{|V|_{\mu_\mathrm{p}}+1}{2}$ for $\mu_\mathrm{c} \in S_\mathrm{c}$ and $\mu_\mathrm{p} \in S_\mathrm{p}$.
Note that the number of solutions is $2^{|S_\mathrm{c}|} = 2^d$ since the assumptions of theorem are $|V|_{\mu_\mathrm{c}} = {\rm even}$ and $|V|_{\mu_\mathrm{p}} = {\rm odd}$.
The number of zero-modes of $\mathcal{D}(G)$ is $\rank \gamma \cdot 2^d$ since $\left( \mathcal{U}^\dagger \mathcal{D}(G)\, \mathcal{U} \right)_{mn}$ contains $\gamma$ matrices.
As a result, the rank of $\mathcal{D}(G)$ is $|V| \cdot \rank \gamma - \rank \gamma \cdot 2^d$.
Therefore, a following equation holds;
\begin{equation}
\begin{split}
    |V| \cdot \rank \gamma - \rank \mathcal{D}(G)
    &= |V| \cdot \rank \gamma - \left( |V| \cdot \rank \gamma - \rank \gamma \cdot 2^M \right)\\
    &= \rank \gamma \cdot 2^{d} 
    = \rank \gamma \cdot \prod_{\mu=1}^{D} \Bigl\{ \beta_0(G_\mu) + \beta_1(G_\mu) \Bigr\} \,.
\end{split}
\end{equation}
And an equation $\frac{\dim \left(\ker \mathcal{D}(G) \right)}{\rank \gamma} = 2^d$ also holds since $\dim \left(\ker \mathcal{D}(G) \right) = |V| \cdot \rank \gamma - \rank \mathcal{D}(G)$.
\end{proof}

For later convenience, we define ``the number of Dirac zero-modes" as the number of zero-modes of $(A_{\rm as})_{\mu} \times \gamma_\mu$ divided by the rank of gamma matrices.
Then, the theorem is rephrased as
{\it the number of Dirac zero-modes of a free and massless fermion is equivalent to the sum of Betti numbers $\beta_{0}+\beta_{1}+...+\beta_{D}$ for lattices (graphs) $G \equiv G_1 \cprod G_2 \cprod \cdots \cprod G_D$ with $G_\mu \in \left\{ D_\mathrm{c}, D_\mathrm{p} \right\}$
with $|V| = {\rm even}$ for $D_\mathrm{c}$ and $|V|={\rm odd}$ for $D_\mathrm{p}$.}

So far, we have restricted our conditions only to the case with $|V| = {\rm even}$ for $D_\mathrm{c}$ and $|V|={\rm odd}$ for $D_\mathrm{p}$.
We now consider more generic cases of $|V|$.
As we have shown, when we have $|V| \not= {\rm even}$ for $D_\mathrm{c}$ or $|V| \not={\rm odd}$ for $D_\mathrm{p}$, the number of zero-modes of $A_{\rm as}$ decreases from the maximum.
Therefore, we have the following theorem:
\begin{theorem}[]
\label{thm6}
For the graphs $G$ constructed a cartesian-product of only the cycle digraph with any vertices and the simple directed path with any vertices, the following equation holds;
\begin{equation}
\label{eq:thm6}
    |V|\cdot \rank \gamma - \rank \mathcal{D}(G) \leq \rank \gamma \cdot \prod_{\mu=1}^{D} \Bigl\{ \beta_0(G_\mu) + \beta_1(G_\mu) \Bigr\}\,,
\end{equation}
where $|V|$ is the number of vertices in $G$ and $\rank \gamma$ is the rank of gamma matrices.
When the graph $G$ is constructed by $d$ cycle digraphs and $(D-d)$ simple directed paths, this equation is rewritten as
\begin{equation}
\label{eq:thm6_v2}
    \frac{\dim \left(\ker \mathcal{D}(G) \right)}{\rank \gamma} \leq 2^{d}\,.
\end{equation}
\end{theorem}

\begin{proof}
The detailed proof is shown in Appendix.~\ref{app:proof_thm6}.
\end{proof}

We obtain two important consequences from these theorems as follows:
\begin{itemize}
    \item 
    There is the maximum number of Dirac zero-modes of $\mathcal{D}(G)$ for $G$,
    \item 
    The maximum number of Dirac zero-modes of $\mathcal{D}(G)$ is two to the power of the number of cycle digraphs constructing the graph $G$.
\end{itemize}

We discuss the relation of Thm.~\ref{thm5} and Thm.~\ref{thm6} to graph topology.
If a $D$-dimensional manifold $M$ is a product space of the circle $T^1$ in $d$ dimensions and the line segment $B^1$ in $D-d$ dimensions, 
the summation of Betti number $\beta_r(M)$ over $0 \leq r \leq D$ is 
\begin{equation}
\label{eq:sum_betti}
    \sum_{r=0}^{D} \beta_r(M) 
    = \sum_{r=0}^{D} \rank H_r(M)
    = 2^d\,,
\end{equation}
as shown in Appendix.~\ref{app:Betti_Kunneth}.
This relation is due to the K\"{u}nneth theorem, which is also referred as in our previous work \cite{Yumoto:2023ums}.
By use of this relation, Thm.~\ref{thm5} and  Thm.~\ref{thm6} are rephrased as
\begin{equation}
\label{eq:thm5-6}
    \frac{\dim \left(\ker \mathcal{D}(G) \right)}{\rank \gamma} \leq \sum_{r=0}^{D} \beta_r(M)\,.
\end{equation}
The equality holds for $|V| = {\rm even}$ for $D_\mathrm{c}$ and $|V|={\rm odd}$ for $D_\mathrm{p}$.
Therefore, the number of Dirac zero-modes of $\mathcal{D}(G)$ is equal to or smaller than the sum of Betti number of the graph $G \equiv G_1 \cprod G_2 \cprod \cdots \cprod G_D$ with $G_\mu \in \left\{ D_\mathrm{c}, D_\mathrm{p} \right\}$.


\section{Implications for Lattice field theory}
\label{sec:LT}

\subsection{Lattice boson}
\label{subsec:LB}

The graph Laplacian defined in Eq.~(\ref{eq:Lap0}) corresponds to the Laplacian in the lattice field theory at least on the hypercubic lattices equivalent to the graphs $G \equiv G_1 \cprod G_2 \cprod \cdots \cprod G_D$ with  $G_\mu \in \left\{ D_\mathrm{c}, D_\mathrm{p} \right\}$ in the previous section.
The topology of this lattice can be a $D$-dimensional torus, $D$-dimensional hyperball(disk) or their cartesian products in $D$-dimensions.
Therefore, the massless action of $D$-dimensional free lattice scalar field $\phi_{n}$ defined on $N^D$ hypercubic-lattice sites is expressed as 
\begin{equation}
S_{\rm b} = -\sum_{n,\mu} {1\over{2}}\phi_{n} (2\phi_{n}-\phi_{n-\hat{\mu}} -\phi_{n+\hat{\mu}}) = -{1\over{2}}\phi L  \phi\,,
\label{eq:boson1}
\end{equation}
with $\phi \equiv (\phi_{1,0,0,...,0}, \phi_{2,0,0,...,0},...,\phi_{N,N,N,...,N})$.
The sum $\sum_{n,\mu}$ is the summation over lattice site $n = (n_{1}, n_{2}, ..., n_{D})$ and $\mu=(1,2,3,...,D)$ with the intervals being $1 \leq n_{\mu} \leq N$. 
$L$ is the graph Laplacian matrix we defined in Eq.~(\ref{eq:Lap0}).
Thus, the spectrum of free and massless lattice boson agrees with that of the graph Laplacian matrix.

Indeed, the equivalence between the lattice scalar operator and the graph Laplacian matrix is not restricted to the above hypercubic lattices. 
In the continuum limit, in which the number of vertices approaches to an infinity with the graph topology being intact, the graph Laplacian results in the continuum Laplacian for an arbitrary lattice or graph.
Thus, the coincidence in Eq.~(\ref{eq:boson1}) holds for generic lattices as
\begin{equation}
S_{\rm b} = \phi {\cal B} \phi= -{1\over{2}}\phi L \phi\,,
\label{eq:boson2}
\end{equation}
where ${\cal B}$ stands for the lattice boson operator.

As we have shown, 
the number of zero modes of the Laplacian matrix is equivalent to 0-th Betti number $\beta_{0}$.
From this fact, we derive the following theorem.
\begin{theorem}[Lattice scalar zero modes]
The number of zero modes of a free and massless lattice scalar operator ${\cal B}$ is equivalent to the 0-th Betti number of the graph (lattice), on which the lattice boson is defined. 
	\begin{equation}
	{\rm dim}({\rm Ker} {\cal B}) \,=\,	{\rm dim}({\rm Ker} L) \,=\, \beta_{0}\,.
	\end{equation}
For any simply connected graphs (lattices), the free boson operator has a single zero mode.	
\end{theorem} 
This theorem holds for any graph (lattice) in any dimensions as long as the lattice boson operator is defined as the graph Laplacian.
The assertion of this theorem for the connected graph is consistent with the results in the work \cite{Ohta:2021xty}.

\subsection{Lattice fermion}
\label{subsec:LF}
For lattice fermions, we restrict ourselves to the graphs (lattices) $G \equiv G_1 \cprod G_2 \cprod \cdots \cprod G_D$ with $G_\mu \in \left\{ D_\mathrm{c}, D_\mathrm{p} \right\}$, whose topology includes a $D$-dimensional torus, $D$-dimensional hyperball(disk) or their cartesian products.

By use of the definition of the anti-symmetrized adjacency matrix in Def.~\ref{def:ASAM}, 
the massless action of $D$-dimensional free naive lattice fermion $\psi_{n}$ on the lattice corresponding to the graph $G \equiv G_1 \cprod G_2 \cprod \cdots \cprod G_D$ with $G_\mu \in \left\{ D_\mathrm{c}, D_\mathrm{p} \right\}$ is expressed as 
\begin{align}
S_{\rm f} &= \sum_{n,\mu} \bar{\psi}_{n}\gamma_{\mu} (\psi_{n-\hat{\mu}} -\psi_{n+\hat{\mu}}) 
\nonumber\\
&= \bar{\psi}{\mathcal D}\psi 
\nonumber\\
&= \bar{\psi} \left[(A_{\rm as})_{\mu} \otimes \gamma_\mu \right]\psi\,.
\end{align}
We here use the standard definition of naive lattice fermion action.
This means that the Dirac operator of the massless, free and naive lattice fermion is equivalent to $(A_{\rm as})_{\mu} \otimes \gamma_\mu$ with $A_{\rm as}$ being the anti-symmetrized adjacency matrix. 
As we have proved, 
the number of zero-modes of the anti-symmetrized adjacency matrix $A_{\rm as}$ is equivalent to the sum of 0-th and 1-st Betti numbers $\beta_{0}+\beta_{1}$ for a cycle digraph $D_\mathrm{c}$ ($T^{1}$ lattice) with $|V| = {\rm even}$ or a simple directed path $D_\mathrm{p}$ ($B^{1}$ lattice) with $|V|={\rm odd}$ (Thm.~\ref{thm4}).
We also proved that the number of Dirac zero-modes of the matrix $(A_{\rm as})_{\mu} \otimes \gamma_\mu$ is equivalent to the sum of Betti numbers $\beta_{0}+\beta_{1}+...+\beta_{D}$ for lattices (graphs) $G \equiv G_1 \cprod G_2 \cprod \cdots \cprod G_D$ with $G_\mu \in \left\{ D_\mathrm{c}, D_\mathrm{p} \right\}$
with $|V| = {\rm even}$ for $D_\mathrm{c}$ and $|V|={\rm odd}$ for $D_\mathrm{p}$ (Thm.~\ref{thm5}).
For any other cases of $|V|$, the number of Dirac zero-modes of $(A_{\rm as})_{\mu} \otimes \gamma_\mu$ is smaller than the sum of Betti numbers $\beta_{0}+\beta_{1}+...+\beta_{D}$ (Thm.~\ref{thm6}).
From these facts, we obtain the following theorem.

\begin{theorem}[Lattice fermion zero modes]
The number of Dirac zero-modes of the free, massless and naive lattice Dirac operator is equivalent to the sum of all the Betti numbers of the graph (lattice) $G \equiv G_1 \cprod G_2 \cprod \cdots \cprod G_D$ with $G_\mu \in \left\{ D_\mathrm{c}, D_\mathrm{p} \right\}$
with $|V| = {\rm even}$ for $D_\mathrm{c}$ and $|V|={\rm odd}$ for $D_\mathrm{p}$, on which the lattice fermion is defined:
	\begin{equation}
		{\rm dim}({\rm Ker} {\mathcal D}) / {\rm rank}\gamma = \sum_{n=0}^{D} \beta_{n}\,.
	\end{equation}
For generic cases of the number of vertices (lattice sites) $|V|$, we have 
\begin{equation}
		{\rm dim}({\rm Ker} {\mathcal D}) / {\rm rank}\gamma \leq \sum_{n=0}^{D} \beta_{n}\,.
	\end{equation}
\end{theorem} 
We speculate that the theorem holds for massive fermions or other lattice fermion formulations since the introduction of mass or the modification of fermion actions never increase the number of Dirac zero-modes of free fermions.

The theorem is consistent with the conjecture on the relation of the number of fermion zero-modes and the sum of Betti numbers proposed in our previous work \cite{Yumoto:2023ums}. It can be well clarified by introducing the Wilson term $\mathcal W$ into the naive fermion and considering the Wilson fermion.
The Wilson fermion action is given by
\begin{align}
S_{\rm W} 
&= \sum_{n,\mu} \bar{\psi}_{n}\gamma_{\mu} (\psi_{n-\hat{\mu}} -\psi_{n+\hat{\mu}}) + {1\over{2}}\sum_{n,\mu}\bar{\psi}_{n}(2\psi_{n}-\psi_{n-\hat{\mu}} -\phi_{n+\hat{\mu}}) 
\nonumber\\
&= \bar{\psi}[{\mathcal D} + {\mathcal W}]\psi
\nonumber\\
&= \bar{\psi} \left[(A_{\rm as})_{\mu} \otimes \gamma_\mu  \,+\, {L\over{2}}\otimes {\bf 1}\right] \psi\,.
\end{align}
It is notable that the added matrix ${\mathcal W}$ corresponding to the Wilson term is proportional to the graph Laplacian matrix as ${\mathcal W} = L/2$.

For the one-dimensional cycle digraph $D_\mathrm{c}$ ($T^{1}$ lattice) with $|V| = {\rm even}$, we have
\begin{align}
\left[ {\mathcal D},\, {\mathcal W}   \right]\,=\,\left[ A_{\rm as},\, {L\over{2}}  \right]\,=\, 0\,,
\end{align}
which means that $A_{\rm as}$ and $L$ are simultaneously diagonalized.
As we have shown, $L$ has a single zero eigenvalue, which is equal to the 0-th Betti number $\beta_{0}=1$.
This zero-eigenvector of $L$ is also one of the zero eigenvectors of $A_{\rm as}$, whose number is $\beta_{0} +\beta_{1}$.
In this sense, {\it the Wilson term ${\mathcal W}=L/2$ works to preserve a single zero-mode associated with $\beta_{0} =1$ out of the two zero-modes of the naive fermion ${\mathcal D} = A_{\rm as}$ associated with $\beta_{0} + \beta_{1} = 2$.}\footnote{Precisely speaking, for the cycle digraph $D_\mathrm{c}$ with $|V| = {\rm even}$,
${\mathcal W}$ is also proportional to the combinatorial Laplacian matrix $L_{1}$.
Thus, we also rephrase that the Wilson term ${\mathcal W}=L_{1}/2$ preserves a single zero-mode associated with $\beta_{1} =1$ out of the two zero-modes.}
This argument is extended to the lattice fermions on higher-dimensional torus lattices:
The Wilson term, which is equivalent to the Laplacian matrix in arbitrary dimensions, works to preserve a single zero-mode associated with $\beta_{0} =1$ out of the $2^d$ zero-modes of the naive fermion.
This is the graph-theoretical reason why the Wilson fermion extracts a single degree of freedom from the multiple doublers.


\section{Summary and Discussion}
\label{sec:SD}
In this paper, we have studied operators in lattice field theory using spectral graph theory and partially proved the conjecture on the relation between the Dirac zero-modes and the Betti numbers of the graph \cite{Yumoto:2023ums}.
We have introduced an ``anti-symmetrized adjacency matrix" $A_{\rm as}$ for a cycle digraph ($T^1$ lattice) and a simple directed path ($B^1$ lattice) and prove that the maximal number of zero eigenvalues of $A_{\rm as}$ is the sum of zeroth and first Betti numbers $\beta_{0} +\beta_{1}$.
It is notable that the graph Laplacian matrix is equivalent to the lattice scalar operator and the Wilson term of lattice fermion, while the anti-symmetrized adjacency matrix and its generalizations to higher dimensions are equivalent to naive Dirac operators.
Based on these facts, we have proved the two theorems:
(i) the number of zero-modes of a free lattice scalar operator agrees with the zero-th Betti number of the graph (lattice) on which the lattice boson is defined, 
(ii) the maximal number of Dirac zero-modes of a free Dirac operator agrees with the sum of all the Betti numbers of the graph (lattice) structured as cartesian products of cycle digraphs ($T^1$ lattice) and simple directed paths ($B^1$ lattice).

The generalization of the theorems to generic cases including the graph with other topologies is one of the urgent problems in the avenue.
For instance, we can construct a $D$-dimensional sphere lattice in a specific manner: It is constructed by putting lids on the top and bottom of the $D$-dimensional cylinder.
$D$-dimensional cylinder is given by a cartesian-product graph as,
\begin{equation}
     G \equiv D_\mathrm{c} \cprod D_\mathrm{p} \cprod \cdots \cprod D_\mathrm{p} \,, 
\end{equation}
where $D_\mathrm{c}$ and $D_\mathrm{p}$ stand for cycle digraphs and simple directed paths.
The only difference between $D$-dimensional cylinder and sphere lattices is the existence of lids at the north and south poles. We have two choices to regard the area surrounded by the edge loop at the top and bottom of the cylinder as a hole or a face, the former of which leads to a $D$-dimensional cylinder and the other of which leads to a $D$-dimensional sphere. 
For a free lattice fermion, this difference is irrelevant since only the plaquette variable composed of link variables lives on faces.
Thus, we can derive the number of zero-modes of lattice fermions on $D$-dimensional spheres just by looking into that on $D$-dimensional cylinders.
We find that the maximal number of zero-modes of free naive lattice fermion on the $D$-dimensional sphere (or the $D$-dimensional cylinder) is two, which is equivalent to the sum of the Betti numbers of $D$-dimensional sphere.
We will prove this fact in details in our upcoming work.

One may ask a question whether we introduce the gauge field into our setups.
Lattice fermion operator with the $U(1)$ background link variable giving a non-zero winding number (topological charge) in two dimensions is regarded as a certain matrix in spectral graph theory.
By use of this fact, we may be able to re-interpret the index theorem connecting the topological charges and the Dirac zero-modes in terms of graph theory.


\begin{acknowledgements}
This work of T. M. is supported by the Japan Society for the Promotion of Science (JSPS) Grant-in-Aid for Scientific Research (KAKENHI) Grant Numbers 23K03425 and 22H05118.
This work of J. Y. is supported by the Sasakawa Scientific Research Grant from The Japan Science Society.
The completion of this work is owed to the discussion in the workshop ‘Lattice field theory and continuum field theory’ at Yukawa Institute for Theoretical Physics, Kyoto University (YITP-W-22-02) and ‘Novel Lattice Fermions and their Suitability for High-Performance Computing and Perturbation Theory’ at Mainz Institute for Theoretical Physics, Johannes Gutenberg University.
\end{acknowledgements}


\appendix
\section{Proofs of the theorems}

\subsection{The proof of theorem.~1}
\label{app:Betti_laplacian}

We will prove that the equation in Eq.~(\ref{eq:lap_betti0}) holds for undirected graphs.
The following proof refers to the proof on eigenvalues of Laplacian matrix \cite{Eigenvalues_Laplacian}.
Firstly, we prove it for a connected graph $G$, whihch is a graph with a single connected component.
The $0$-th Betti number $\beta_0(G)$ is $\beta_0(G) = 1$ since the graph has a single connected component.
Due to this fact, what we need to do is to show $|V| - \rank L(G) = \beta_0(G) = 1$.
For a vector $\bm{v} \in \mathbb{C}^{|V|}$, where $\displaystyle \bm{v} = \sum_{i=1}^{|V|} x_i \bm{e}_i$ with $\bm{e}_i$ is the standard basis, $\bm{v}^\dagger L \bm{v}$ is 
\begin{equation}
    \bm{v}^\dagger L(G) \bm{v} 
    = \bm{v}^{\dagger} \left( D(G) - A(G) \right) \bm{v}
    = \bm{v}^\dagger D(G) \bm{v} - \bm{v}^\dagger A(G) \bm {v} \,,
\end{equation}
where $D(G)$ and $A(G)$ are the degree matrix of $G$ and the adjacency matrix of $G$ respectively.
Based on the definitions of degree matrix and adjacency matrix, 
\begin{equation}
\begin{split}
    \bm{v}^\dagger L(G) \bm{v} 
    &= \sum_{i=1}^{|V|} \deg(v_i) |x_i|^2 
        - \sum_{\left\{i,j\right\} \in E} 
        \left( \bar{x}_ix_j + \bar{x}_jx_i \right) \\
    &= \sum_{\left\{i,j\right\} \in E} 
        \left( |x_i|^{2} + |x_j|^{2} \right)
        - \sum_{\left\{i,j\right\} \in E} 
        \left( \bar{x}_ix_j + \bar{x}_jx_i \right) \\
    &= \sum_{\left\{i,j\right\} \in E} 
        \left( \bar{x}_i - \bar{x}_j \right)
        \left( x_i - x_j \right) \\
    &= \sum_{\left\{i,j\right\} \in E} 
        \left| x_i - x_j \right|^{2}
\end{split}\,,
\end{equation}
where $\left\{i,j\right\}$ stands for an edge between one vertex $v_i$ and other vertex $v_j$.
Furthermore, $V$ and $E$ are the set of vertices and the set of edges in $G$.

If the vector $\bm{v}$ is a zero-mode (zero-eigenvector) of the Laplacian $L(G)$, $\bm{v}^\dagger L(G) \bm{v}$ satisfies the following equation,
\begin{equation}
    \bm{v}^\dagger L(G) \bm{v} 
    = \sum_{\left\{i,j\right\} \in E} 
        \left| x_i - x_j \right|^{2}
    = 0 \,.
\end{equation}
From this and $\left| x_i - x_j \right|^{2} \geq 0$, $x_i = x_j$ for any edge $\left\{ i,j \right\} \in E$ is obtained.
Since we consider a connected graph, the components of the zero-mode $\bm{v}$ must satisfy $x_1 = x_2 = \cdots = x_{|V|}$.
If there is $x_i \neq x_j$ in components of zero-mode $\bm{v}$, the graph can be divided into $G_1$ and $G_2$ such that there is no edges between $G_1$ and $G_2$.
But, this is inconsistent with the fact that the graph is a connected graph.
We then have that the zero-mode $\bm{v}$ for Laplacian $L(G)$ is unique and explicitly written as $\bm{v} = \alpha \sum_{i=1}^{|V|} \bm{e}_i$ with $\alpha \in \mathbb{C}$.
Hence, $|V| - \rank L(G) = \beta_0(G) = 1$ holds for a connected graph $G$ since the rank of Laplacian is $L(G) = |V|-1$.


Secondly, we prove it for a graph with $k$ connected components.
The $0$-th Betti number for this graph is $\beta_{0}(G) = k$ since the $0$-th Betti number is equal to the number of connected components.
What we have to do is to prove $|V| - \rank L(G) = \beta_{0}(G) = k$.
The graph can be divided into $k$ connected graphs such that there are no edges between each two of connected graphs.
They are denoted as $G_\nu$ for $\nu \in \left\{ 1,2,\cdots, k\right\}$.
Furthermore, the number of vertices of $G_\nu$ is denoted as $|V|_{G_\nu}$ and $\sum_{\nu=1}^{k} |V|_{G_\nu} = |V|$. 
Then, a Laplacian $L(G)$ is a block matrix as
\begin{equation}
    L(G) = 
    \begin{pmatrix}
        L(G_1)\\
        &L(G_2)\\
        &&\ddots\\
        &&&L(G_k)
    \end{pmatrix}\,,
\end{equation}
since there are no edges between each two of connected graphs.
The rank of each matrix $L(G_\nu)$ is $\rank L(G_\nu) = |V|_{G_\nu}-1$ because each matrix $L(G_\nu)$ is a Laplacian for a connected graph $G_\nu$.
Hence, $|V| - \rank (G) = \beta_{0}(G) = k$ holds for the graph $G$ with $k$ connected components since the rank of the Laplacian $L(G)$ is given by $\rank L(G) = \sum_{\nu=1}^{k} \rank L(G_\nu) = |V| - k$.
We now conclude that $|V| - \rank (G) = \beta_{0}(G)$ holds for generic undirected graphs.

Next, we will prove that the equation in Eq.~(\ref{eq:lap_betti1}) holds for undirected graphs.
We consider a graph with $k$ connected components.
The Euler characteristic $\chi(G)$ of the graph $G$ is $\chi(G) = \beta_0(G) - \beta_1(G) = |V| - |E|$. 
The $0$-th Betti number is then expressed as $\beta_0(G)= |V| - \rank L(G)$ by use of the above proof we have proved.
Therefore, $|E| - \rank L(G) = \beta_1 (G)$ holds for any graphs since $\chi(G) = |V| - \rank L(G) - \beta_1(G) = |V| - |E|$.


\subsection{The proof of theorem.~6}
\label{app:proof_thm6}
We will prove Thm.~\ref{thm6}. 
We consider four digraphs: the cycle digraph with even vertices $D_\mathrm{c}^{\mathrm{(even)}}$, the one with odd vertices $D_\mathrm{c}^{\mathrm{(odd)}}$, the simple directed path with even vertices $D_\mathrm{p}^{\mathrm{(even)}}$, and the one with odd vertices $D_\mathrm{p}^{\mathrm{(odd)}}$.
The graph $G$ is constructed a cartesian-product of these digraphs.
The number of the cycle digraphs is $d$ and the number of the simple directed paths is $D-d$.
Then, the right side in Eq.~(\ref{eq:thm6}) is $\Disp{\rank \gamma \cdot \prod_{\mu=1}^{D} \Bigl\{ \beta_0(G_\mu) + \beta_1(G_\mu) \Bigr\} = \rank \gamma \cdot 2^d}$.
Due to this fact, what we need do is just to show
\begin{equation}
    |V|\cdot \rank \gamma - \rank \mathcal{D}(G) 
    \leq \rank \gamma \cdot 2^d\,.
\end{equation}
Similarly to the proof of Thm.~\ref{thm5}, the eigenvalues of $\mathcal{D}(G)$ defined in Eq.~(\ref{eq:D_matrix}) is
\begin{equation}
\begin{split}
    &\left( \mathcal{U}^\dagger \mathcal{D}(G)\, \mathcal{U} \right)_{mn}\\
    &= 2i \left\{
        \sum_{\mu_\mathrm{c}^{\mathrm{even}} \in S_\mathrm{c}^{\mathrm{even}}} \gamma_{\mu_\mathrm{c}^{\mathrm{even}}}
        \sin \left( \frac{2\pi(m_{\mu_\mathrm{c}^{\mathrm{even}}}-1)}{|V|_{\mu_\mathrm{c}^{\mathrm{even}}}} \right)
        +
        \sum_{\mu_\mathrm{c}^{\mathrm{odd}} \in S_\mathrm{c}^{\mathrm{odd}}} \gamma_{\mu_\mathrm{c}^{\mathrm{odd}}}
        \sin \left( \frac{2\pi(m_{\mu_\mathrm{c}^{\mathrm{odd}}}-1)}{|V|_{\mu_\mathrm{c}^{\mathrm{odd}}}} \right)
        \right.\\
    &\qquad\left.+
        \sum_{\mu_\mathrm{p}^{\mathrm{even}} \in S_\mathrm{p}^{\mathrm{even}}} \gamma_{\mu_\mathrm{p}^{\mathrm{even}}} 
        \cos \left( \frac{m_{\mu_\mathrm{p}^{\mathrm{even}}} \pi}{|V|_{\mu_\mathrm{p}^{\mathrm{even}}} +1} \right) 
        +
        \sum_{\mu_\mathrm{p}^{\mathrm{odd}} \in S_\mathrm{p}^{\mathrm{odd}}} \gamma_{\mu_\mathrm{p}^{\mathrm{odd}}} 
        \cos \left( \frac{m_{\mu_\mathrm{p}^{\mathrm{odd}}} \pi}{|V|_{\mu_\mathrm{p}^{\mathrm{odd}}} +1} \right) 
    \right\} \delta_{mn}\,,
\end{split}
\end{equation}
where four sets are defined as
\begin{equation}
\begin{gathered}
    S_\mathrm{c}^{\mathrm{even}} \equiv 
    \left\{ \mu \setmid G_\mu = D_\mathrm{c}^{\mathrm{(even)}} \right\}
    ,\quad
    S_\mathrm{c}^{\mathrm{odd}} \equiv
    \left\{ \mu \setmid G_\mu = D_\mathrm{c}^{\mathrm{(odd)}} \right\}
    ,\\
    S_\mathrm{p}^{\mathrm{even}} \equiv
    \left\{ \mu \setmid G_\mu = D_\mathrm{p}^{\mathrm{(even)}} \right\}
    ,\quad
    S_\mathrm{p}^{\mathrm{odd}} \equiv
    \left\{ \mu \setmid G_\mu = D_\mathrm{p}^{\mathrm{(odd)}} \right\}\,.
\end{gathered}
\end{equation}
respectively.
And these sets satisfy $|S_\mathrm{c}^{\mathrm{even}}| + |S_\mathrm{c}^{\mathrm{odd}}| = d$ and $|S_\mathrm{p}^{\mathrm{even}}| + |S_\mathrm{p}^{\mathrm{odd}}| = D-d$.
Then, the conditions for the diagonal matrix to have zero-modes are
\begin{equation}
\begin{split}
    &\sin \left( \frac{2\pi(m_{\mu_\mathrm{c}^{\mathrm{even}}}-1)}{|V|_{\mu_\mathrm{c}^{\mathrm{even}}}} \right)
    =
    \sin \left( \frac{2\pi(m_{\mu_\mathrm{c}^{\mathrm{odd}}}-1)}{|V|_{\mu_\mathrm{c}^{\mathrm{odd}}}} \right)\\
    &= 
    \cos \left( \frac{m_{\mu_\mathrm{p}^{\mathrm{even}}} \pi}{|V|_{\mu_\mathrm{p}^{\mathrm{even}}} +1} \right)
    = 
    \cos \left( \frac{m_{\mu_\mathrm{p}^{\mathrm{odd}}} \pi}{|V|_{\mu_\mathrm{p}^{\mathrm{odd}}} +1} \right)
    = 0\,,
\end{split}
\end{equation}
where we used linearly independence of $\gamma$-matrices.
The solutions of this equation are
\begin{equation}
    m_{\mu_\mathrm{c}^{\mathrm{even}}} = 1\,, \quad
    m_{\mu_\mathrm{c}^{\mathrm{odd}}} = 0\,, \quad
    m_{\mu_\mathrm{p}^{\mathrm{odd}}} = \frac{|V|_{\mu_\mathrm{p}^{\mathrm{odd}}} +1}{2}\,
\end{equation}
or
\begin{equation}
    m_{\mu_\mathrm{c}^{\mathrm{even}}} = \frac{|V|_{\mu_\mathrm{c}^{\mathrm{even}}}}{2}+1\,, \quad
    m_{\mu_\mathrm{c}^{\mathrm{odd}}} = 0\,, \quad
    m_{\mu_\mathrm{p}^{\mathrm{odd}}} = \frac{|V|_{\mu_\mathrm{p}^{\mathrm{odd}}} +1}{2}\,.
\end{equation}
Note that there are no solutions if $S_\mathrm{p}^{\mathrm{even}} \neq \emptyset$.
Hence, $|V| - \rank \mathcal{D}(G)$ is obtained as
\begin{equation}
    |V| \cdot \rank \gamma - \rank \mathcal{D}(G) = 
    \begin{cases}
        \rank \gamma \cdot 2^{|S_\mathrm{c}^{\mathrm{even}}|} \cdot 1^{|S_\mathrm{c}^{\mathrm{odd}}|} \cdot 1^{|S_\mathrm{p}^{\mathrm{odd}}|} 
            &S_\mathrm{p}^{\mathrm{even}} = \emptyset\\
        0   &S_\mathrm{p}^{\mathrm{even}} \neq \emptyset
    \end{cases}\,.
\end{equation}
With $|S_\mathrm{c}^{\mathrm{even}}| = d$ and $|S_\mathrm{p}^{\mathrm{odd}}| = D-d$, $|V| - \rank \mathcal{D}(G)$ gets the maximum, which is $\rank \gamma \cdot 2^d$.
Therefore, the following inequality has been proved
\begin{equation}
    |V| \cdot \rank \gamma - \rank \mathcal{D}(G) \leq \rank \gamma \cdot 2^d
    = \rank \gamma \cdot \prod_{\mu=1}^{D} \Bigl\{ \beta_0(G_\mu) + \beta_1(G_\mu) \Bigr\}\,.
\end{equation}

As seen from this proof, we can also discuss the minimum number of zero-modes of $\mathcal{D}(G)$.
It is expressed as
\begin{equation}
    \min \left[  |V| \cdot \rank \gamma - \rank \mathcal{D}(G) \right]
    = \begin{cases}
        \rank \gamma   
            &S_\mathrm{p}^{\mathrm{even}} = \emptyset\\
        0
            &S_\mathrm{p}^{\mathrm{even}} \neq \emptyset
    \end{cases} \,,
\end{equation}
or 
\begin{equation}
    \min \left[ 
        \frac{\dim \left(\ker \mathcal{D}(G) \right)}{\rank \gamma}
    \right]
    = \begin{cases}
        1
            &S_\mathrm{p}^{\mathrm{even}} = \emptyset\\
        0
            &S_\mathrm{p}^{\mathrm{even}} \neq \emptyset
    \end{cases}\,.
\end{equation}

\subsection{Graph topology and Betti numbers}
\label{app:Betti_Kunneth}

We will prove the equation in Eq.~(\ref{eq:sum_betti}).
We assume that a $D$-dimensional manifold $M$ is a product space of the circle $T^1$ in $d$ dimensions and the line segment $B^1$ in $D-d$ dimensions. That is expressly written as $M = M_1 \times M_2 \times \cdots \times M_D$, where $M_\mu \in \left\{T^1, B^1 \right\}$ for $\mu \in \left\{ 1,2,\cdots,D\right\}$. Furthermore, the two sets regarding $\mu$, which are defined as $S_\mathrm{circ} = \left\{ \mu \setmid M_\mu = T^1 \right\}$ and $S_\mathrm{line} = \left\{ \mu \setmid M_\mu = B^1 \right\}$, satisfy $|S_\mathrm{circ}| = d$ and $|S_\mathrm{line}| = D - d$, respectively.
From K\"{u}nneth theorem, $r$-th homology of the manifold $M$ is written down as
\begin{equation}
\label{app:Kunneth}
    H_r (M) \cong
    \bigoplus_{ \sum_{\mu=1}^{D}r_{\mu} = r}
    \left(
        \bigotimes_{\mu=1}^{D} H_{r_\mu} (M_\mu)
    \right)\,,
\end{equation}
where $H_r(M)$ stands for the $r$-th homology of manifold $M$.
Now, we can restrict $r_\mu$ in $H_{r_\mu}(T^1)$ to $0\leq r_\mu \leq 1$ 
since the homology of $T^1$ is $H_{r_\mu} (T^1) = \mathbb{Z}$ for $r_\mu=0,\,1$, otherwise $H_{r_\mu} (T^1) = 0$.
On the one hands, $r_\mu$ in $H_{r_\mu}(B^1)$ can be restricted to $0$ since the homology of $B^1$ is $H_{0} (B^1) = \mathbb{Z}$ otherwise $H_{r_\mu} (B^1) = 0$.
Accordingly, the direct-product of $H_{r_\mu} (M_\mu)$ over $1 \leq \mu \leq D$ is
\begin{equation}
    \bigotimes_{\mu=1}^{D} H_{r_\mu} (M_\mu) 
    = \bigotimes_{\mu=1}^{D} \mathbb{Z}
    = \mathbb{Z}\,,
\end{equation}
since $\mathbb{Z} \ot \mathbb{Z} = \mathbb{Z}$.
And the summation of $r_\mu$ in Eq.~(\ref{app:Kunneth}) is rewritten as $\sum_{\mu=1}^{D} r_\mu = \sum_{\mu \in S_\mathrm{circ}} r_\mu + \sum_{\mu \in S_\mathrm{line}} r_\mu= \sum_{\mu \in S_\mathrm{circ}} r_\mu = r$ since $r_\mu$ for $\mu \in S_\mathrm{line}$ is restricted to $r_\mu = 0$.
As a result, Eq.~(\ref{app:Kunneth}) is rewritten as
\begin{equation}
    H_r (M) \cong
    \bigoplus_{\sum_{\mu \in S_\mathrm{circ}} r_\mu = r} \mathbb{Z}\,.
\end{equation}
The summation of $r_\mu$ shows that the right side in above equation is equal to ${}_{d}C_{r}$ dimensional integer space, that is $H_r (M) \cong \mathbb{Z}^{{}_{d}C_{r}}$.
Because the number of combination of $r_\mu$ that are $ \sum_{\mu \in S_\mathrm{circ}} r_\mu = r$ is ${}_{d}C_{r}$.
Hence, the rank of $r$-th homology is $ \rank H_r (M) = \rank \mathbb{Z}^{{}_{d}C_{r}} = {}_{d}C_{r}$.
Since $r$-th Betti number is the rank of $r$-th homology, the summation of Betti numbers $\beta_r (M)$ over $0 \leq r \leq D$ is obtained as
\begin{equation}
    \sum_{r=0}^D \beta_r (M) 
    = \sum_{r=0}^D \rank H_r (M)
    = \sum_{r=0}^D {}_{d}C_{r}
    = \sum_{r=0}^d {}_{d}C_{r} + \sum_{r=d+1}^D 0
    = 2^d \,,
\end{equation}
where we use the property of ${}_{r}C_{d}$ (${}_{r}C_{d} = 0$ for $r > d$).
Therefore, Eq.~(\ref{eq:sum_betti}) holds for the $D$-dimensional manifold $M$, which is a product space of the circle $T^1$ in $d$ dimensions and the line segment $B^1$ in $D-d$ dimensions.

\bibliographystyle{utphys}
\bibliography{./QFT,./refs,./math}

\end{document}